%% file: main.tex
\documentclass[a4paper, 11pt]{article}

\usepackage{graphicx} 
\usepackage{pgf}
\usepackage{times} 
\usepackage{stmaryrd}

\usepackage{amsmath} 
\usepackage{amssymb}  
\usepackage{amsthm}
\usepackage{cite}

\title{Stabilization of finite-energy grid states of a quantum harmonic oscillator by reservoir engineering with two dissipation channels}

\author{Rémi Robin$^{1}$ \and Pierre Rouchon$^{1}$ \and Lev-Arcady Sellem$^{2}$
\thanks{*This project has received funding from the European Research Council (ERC) under the European Union’s Horizon 2020 research and innovation programme (grant agreement No. [884762]) and Plan France 2030 through
the project ANR-22-PETQ-0006..}
\thanks{$^{1}$Laboratoire  de Physique de l’Ecole normale sup\'{e}rieure, Mines Paris-PSL, Inria,  ENS-PSL, Universit\'{e} PSL, CNRS, Sorbonne Universit\'{e},  Paris, France.}%
\thanks{$^{2}$Institut Quantique and Département de Physique, Université de Sherbrooke, Sherbrooke J1K 2R1 QC, Canada.}%
}

\input{symbolsmacros.tex}

\begin{document}

\maketitle

\begin{abstract}

We propose and analyze an experimentally accessible Lindblad master equation for a quantum harmonic oscillator, simplifying a previous proposal to alleviate implementation constraints. It approximately stabilizes periodic grid states introduced in 2001 by Gottesman, Kitaev
and Preskill (GKP), with applications for quantum error correction and quantum metrology.
We obtain explicit estimates for the energy of
the solutions of the Lindblad master equation.
We estimate the convergence rate to the codespace when stabilizing a GKP qubit, and numerically study the effect of noise.
We then present simulations illustrating how a modification of parameters allows preparing states of metrological interest in steady-state.

\end{abstract}
\section{INTRODUCTION}
Bosonic codes, such as
the binomial
\cite{MichaelSilveriBrierleyEtAl2016},
cat \cite{cochraneMacroscopicallyDistinctQuantumsuperposition1999} or Gottesman-Kitaev-Preskill (GKP) \cite{gottesmanEncodingQubitOscillator2001} code,
aim at drastically alleviating the hardware requirements of quantum error correction (QEC). In a nutshell, they introduce the redundance necessary for error correction by encoding information in exotic states in the infinite-dimensional Hilbert space of a single harmonic oscillator, instead of using the exponentially large Hilbert space of a collection of qubits.
In particular, several experiments in superconducting circuits
\cite{campagne-ibarcqQuantumErrorCorrection2020,NeeveGKP2022,sivakRealtimeQuantumError2023a,brockQuantumErrorCorrection2024,lachance-quirion_2024}
and trapped ions \cite{fluhmannEncodingQubitTrappedion2019,NeeveGKP2022,matsosRobustDeterministicPreparation2023}
demonstrated generation and stabilization of GKP states. They rely on discrete-time control protocols using an auxiliary qubit to control the state of a harmonic oscillator. This strategy is typically limited by the propagations of errors affecting the ancilla to the encoded qubit.
Several methods have been proposed where the codespace of a GKP qubit is instead autonomously stabilized in continuous time by carefully engineering its coupling to a strongly damped environment
\cite{sellemDissipativeProtectionGKP2025,sellemStabilityDecoherence2023,sellemExponentialConvergenceDissipative2022a,nathanSelfCorrectingGottesmanKitaevPreskillQubit2025,geierSelfcorrectingGKPQubit2024}, promising order of magnitude improvements in the achievable logical lifetimes.
These theory proposals typically place stringent constraints on hardware parameters and/or control capabilities and have yet to be demonstrated experimentally.

Here, we leverage a symmetry in the GKP code to introduce a simplification of the protocol proposed in \cite{sellemDissipativeProtectionGKP2025}, and study the performance of this new protocol both for the stabilization of a qubit and for the generation of GKP states for quantum metrology.
We obtain \emph{a priori} estimates on the solutions of the proposed Lindblad equation, in the form of explicit energy bounds. We then exploit the periodicity of so-called stabilizers operators that define the GKP subspace to relate the convergence to that subspace to spectral properties of a singular differential operator on periodic functions, and obtain explicit estimates on its low-lying spectrum.
Finally, we turn to numerical simulations to study the robustness of the stabilization in presence of photon loss.
\\

The paper is organized as follows. Sec. \ref{sec-defs} gathers basic definitions and notations.
Sec. \ref{sec-main-eq} presents the proposed simplification to the stabilizing dynamics studied in \cite{sellemDissipativeProtectionGKP2025},
as well as numerical simulations suggesting it is sufficient to stabilize GKP states. We then turn  to mathematically characterizing this stabilization: in Sec. \ref{sec-stab}, we derive \emph{a priori} estimates on solutions; in  Sec. \ref{sec-log}, we derive explicit convergence rates of logical observables defining the codespace in the case of a GKP qubit; and in Sec. \ref{sec-noise} we study the impact of noise in the form of photon loss.

In Sec. \ref{sec-qunaught} we study a slight modification of parameters that allows stabilizing a single state instead of the two-dimensional state space of a qubit.
In Sec. \ref{sec-implem} we shortly discuss possible physical platforms for implementing this stabilizing dynamics.
The Appendix  provides additional details on calculations omitted in the main text.

\section{Gottesman-Kitaev-Preskill spaces}
\label{sec-defs}
Denote $\oq$ and $\op$ the position and momentum operators of a quantum harmonic oscillator, satisfying $[\oq, \op] = i$.
From Glauber identity, for any $\eta>0$ such that $d := \eta^2/ 2\pi \in \mathbb N$, the periodic operators $e^{\pm i\eta \oq}$ and $e^{\pm i\eta \op}$ commute, and generate the so-called stabilizer group $\{ e^{i m \eta \oq} e^{i n \eta \op}\, , (m,n) \in \mathbb Z^2\}$.
Recall that in the $q$-representation, $\op = -i \frac{d}{dq}$ so that $e^{\pm i\eta p}$ is a translation operator on $q$. Then, for any wavefunction $|\psi\rangle = (\psi(q))_{q\in\mathbb R}$, we have:
\begin{equation}
    e^{ im \eta \oq} e^{i n \eta \op} |\psi\rangle =
    \left( e^{i m \eta q} \psi(q+n\eta) \right)_{q\in\mathbb R}.
    \label{eq-trans-sol}
\end{equation}
Solving for $+1$ eigenstates of the operator $e^{ im \eta \oq} e^{i n \eta \op}$ in Eq. \eqref{eq-trans-sol}, one formally defines the GKP subspace of a qudit as the $d$-dimensional joint eigenspace of the stabilizers associated to the eigenvalue $+1$, spanned by the Dirac combs
$\psi_k(q) := \sum_{n\in\mathbb Z} \delta(q-(nd+k)\frac{2\pi}\eta)$,
where $k\in \{0,\ldots,d-1\}$ and $\delta$ stands for the Dirac distribution.
In particular, when $d=2$, one can encode a qubit in the subspace spanned by the two GKP states, protected from the effect of local noise by the distance between the supports of the two logical states in phase space.
When $d=1$, on the other hand, the GKP subspace is $1$-dimensional (the corresponding Dirac comb is also known as a qunaught); while it cannot encode logical information, it has been considered for metrological applications
\cite{Duivenvoorden2017,valahuQuantumenhancedMultiparameterSensing2025,labarcaQuantumSensingDisplacements2026}.

The GKP states defined above are not square integrable and thus not valid wavefunctions; several regularized versions, called finite-energy GKP states, coexist in the literature \cite{matsuura2020equivalence}. We follow \cite{gottesmanEncodingQubitOscillator2001} and introduce a regularizing operator $\oE_\epsilon := e^{-\frac\epsilon2(\oq^2+\op^2)}$ parametrized by $\epsilon>0$. The finite-energy GKP states are then related to their infinite-energy counterpart through
\begin{equation}
    |\psi_{k,\epsilon}\rangle := \oE_\epsilon |\psi_k\rangle,
\end{equation}
see e.g. \cite{sellemExponentialConvergenceDissipative2022a} for explicit expressions of the corresponding wavefunctions.
Note that the energy of these states is defined as the expectation value of the photon-number operator $\oN = \oa^\dagger \oa$ (with $\oa = (\oq - i\op)/\sqrt2$ the so-called annihilation operator). For $\epsilon \ll 1$, it grows as $\langle \oN \rangle \propto 1/\epsilon$ on GKP states.

\section{TWO-DISSIPATOR DYNAMICS}
\label{sec-main-eq} 
In \cite{sellemDissipativeProtectionGKP2025},
a Lindblad type dynamics was studied to stabilize a finite-energy square GKP code. It involves 
$4$ dissipators, inspired from the stabilizers of the logical code,
\begin{equation}
\begin{aligned}
    \oL_1 &= \sin(\etasq \oq) + i\epsilon\cos(\etasq \oq) \op,\\
    \oL_2 &= \sin(\etasq \op) - i\epsilon\cos(\etasq \op) \oq,\\
    \oL_3 &=   \cos(\etasq \oq) - i\epsilon\sin(\etasq \oq) \op  - e^{ \epsilon\etasq/2}\,\oid ,\\
    \oL_4 &= \cos(\etasq \op) + i\epsilon\sin(\etasq \op) \oq - e^{ \epsilon\etasq/2}\,\oid,
\end{aligned}
\label{eq-prx4}
\end{equation}
where $\etasq = \sqrt{2\pi d}$ is the lattice constant of a square GKP qudit for $d\in \mathbb N^*$: the choice $\etasq = 2\sqrt\pi$ ($d=2$) stabilizes a qubit, while $\etasq = \sqrt{2\pi}$ ($d=1$) stabilizes a single state.

Here, we study the properties of the following modification of this scheme: we consider only the first two dissipators, but replace $\etasq$ with $\eta = \etasq/2$, for the case of a qubit ($\eta = \sqrt\pi$) and a qunaught ($\eta  = \sqrt{\pi/2}$), leading to the Lindblad equation
\begin{equation}
    \frac{d}{dt}\orho =  \mathcal L(\orho) := D[\oM_1](\orho) +  D[\oM_2](\orho)
    \label{eq-main-lindblad}
\end{equation}
with $D[\oM](\orho) := \oM \orho \oM^\dagger - \frac12\oM^\dagger\oM\orho - \frac12\orho\oM^\dagger\oM$ and
\begin{align}
    \oM_1 &= \sin(\eta \oq) + i\epsilon\cos(\eta \oq) \op\\
    \oM_2 &= \sin(\eta \op) - i\epsilon\cos(\eta \op) \oq.
\end{align}

The motivation behind these modifications can be understood through the following naive intuition:
understanding the terms in $\epsilon$ as perturbations, considering Lindblad operators as constraints applied to a system, and ignoring non-commutativity of operators for the sake of intuition, the constraint $\sin(\theta) = 0$ is equivalent to $\sin(2\theta) = 0; \; \cos(2\theta)=1$.
As we'll briefly explain in Sec. \ref{sec-implem},
both the reduction of $\eta$ and the number of dissipators could significantly alleviate experimental implementation constraints.

\begin{figure}
    \includegraphics[width=\columnwidth]{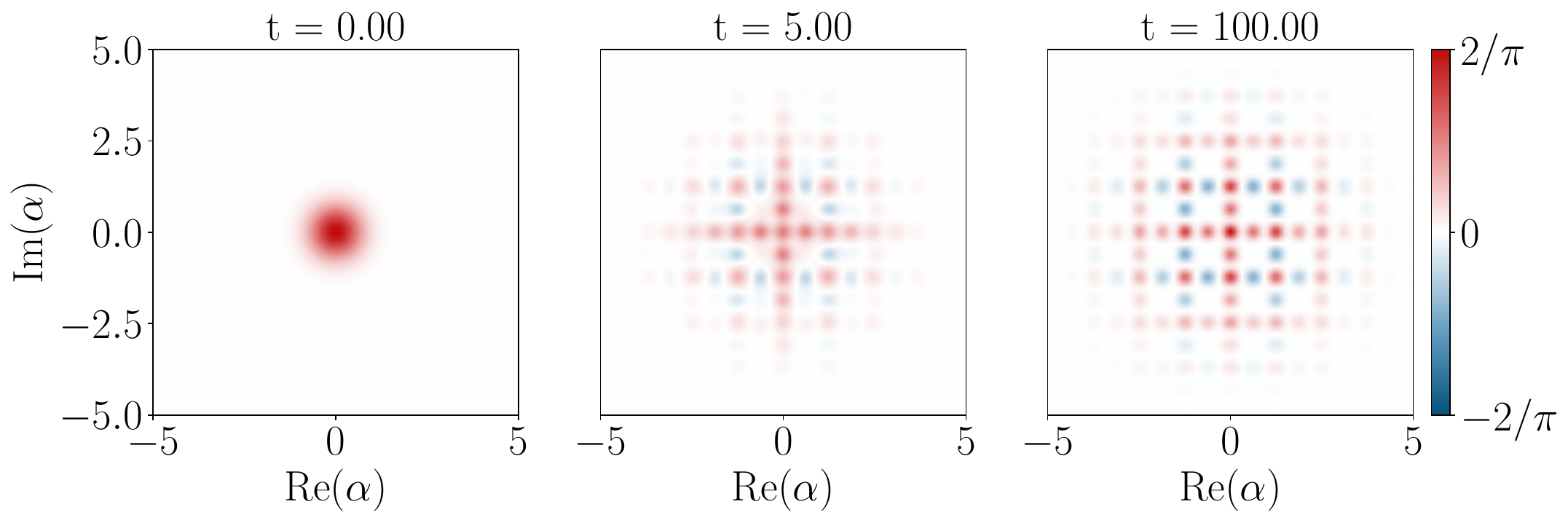}
    \\
    \includegraphics[width=\columnwidth]{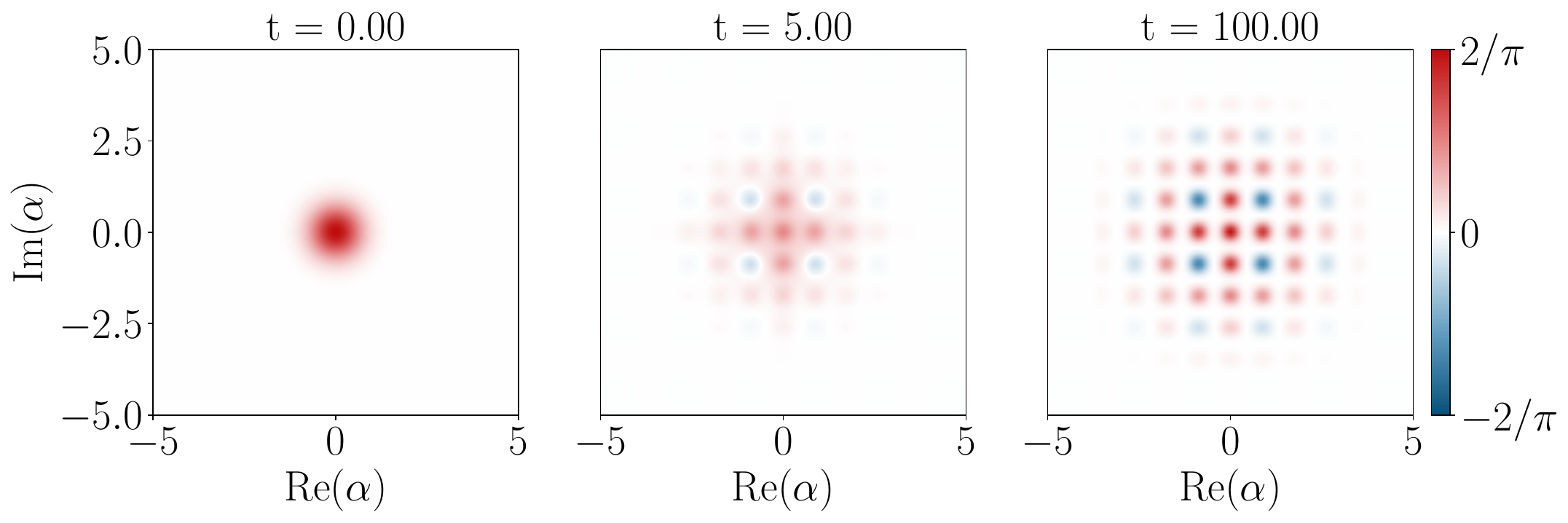}
    \caption{Long-time simulations of Eq. \eqref{eq-main-lindblad}; initialized in vacuum $\orho_0 = |0\rangle\langle 0|$, for $\epsilon=0.15$. {\bf Top.} With $\eta = \sqrt\pi$, we observe the state converge to a GKP qubit state (for this initial state, the steady state is close to the so-called magic state \cite{bravyiUniversalQuantumComputation2005b} $\cos(\pi/8)|+Z\rangle + \sin(\pi/8)|-Z\rangle$ inside the GKP code, with fidelity $>92\%$). {\bf Bottom.} With $\eta = \sqrt{\pi/2}$ there is a single steady state, also known as a GKP qunaught.}
    \label{eq-conv-ideal}
\end{figure}

Fig. \ref{eq-conv-ideal} shows long-time simulations of Eq. \eqref{eq-main-lindblad} starting from the vacuum state $\orho_0 = |0\rangle\langle 0|$, which is a natural initial state as preparing vacuum amounts to letting the system converge to its ground state (provided its environment is cold enough to approximate its temperature as being $0$). It illustrates how this simplified dynamics already allows for the stabilization of GKP qubit or qunaught states.
In the  rest of the paper, we now provide analytical results supporting this property and numerically study the effect of noise, in the form of photon loss.

\section{A priori estimates}
\label{sec-stab}
Direct calculations, adapted from the methodology in \cite{sellemStabilityDecoherence2023}, allow deriving explicit bounds on the energy, defined
as the expectation value $\langle \oN \rangle = \tr(\oN\orho)$ of the photon
number operator $\oN =\oa^\dagger \oa$, along trajectories of a density operator governed by Eq. \eqref{eq-main-lindblad}.
Here, we obtain a priori estimates by formal computations,
led as if the dimension of the underlying Hilbert space
were finite. We plan to exploit these estimates for a fully
rigorous mathematical analysis in future publications.
\begin{theorem}
    For $\eta>0$, there exists $\overline \epsilon \in (0,2/\eta)$ such that for all $\epsilon<\overline \epsilon$ and $\orho_0$ with $\tr(\oN \orho_0) < \infty$, for all $t\geq 0$:

\begin{equation}
\frac{d}{dt}\tr(\oN \orho_t) \leq -\lambda(\epsilon,\eta) \tr(\oN \orho_t) + \mu(\epsilon,\eta)
\label{eq-bound-N}
\end{equation}
with $\lambda>0, \mu>0$ some constants depending on $\epsilon, \eta$.
In particular, the energy along trajectories remains bounded by $\max \left(\tr(\oN \orho_0), \mu(\epsilon,\eta)/\lambda(\epsilon,\eta) \right)$ at all times. Moreover, $\lambda(\epsilon,\eta)$ can be chosen arbitrarily close below $2\epsilon\eta (1-\epsilon\eta/2)$.
\end{theorem}\vspace{0.5em}
The proof of this theorem is mostly technical and deferred to the Appendix.
Note that this result, relating only to the stability of the dynamics, is expressed for arbitrary values $\eta>0$, although only the values $\eta\in\{\sqrt\pi, \sqrt{\pi/2}\}$ are considered in the rest of the paper.

\section{Convergence of logical observables}
\label{sec-log}

Let us consider the case of a qubit ($d=2$, corresponding to $\eta = \sqrt\pi$). Two notions of convergence are of interest. On the one hand, any state should converge to the codespace, that is the two-dimensional subspace spanned by GKP states. On the other hand, adding spurious noise processes on top of the stabilizing dynamics can lead to an evolution inside the logical subspace:  in the infinite time limit, any state decoheres to a unique mixed state, and the convergence rate of this evolution inside the codespace quantifies the timescale on which logical information decoheres.

To quantify convergence to the codespace, let us recall that GKP states are, up to small errors due to normalization, close to $+1$ eigenstates of the GKP stabilizers. Instead of solving for the evolution of states through the Lindblad equation, we can solve for the evolution of periodic observables through its adjoint, \emph{i.e.} in the so-called Heisenberg picture.

\input{spectralgap.tex}

\section{Effect of photon loss}
\label{sec-noise}
To study the robustness to noise, we modify the Lindblad equation to add a term modeling photon loss at rate $\kappa>0$:
\begin{equation}
    \frac{d}{dt} \orho = D[\oM_1](\orho) + D[\oM_2](\orho) + \kappa D[\oa](\orho).
    \label{eq-stab-loss}
\end{equation}
Note that without the stabilization (dissipators in $\oM_1$ and $\oM_2$), any state converges exponentially to the vacuum state $|0\rangle$.
When the stabilization is on, the solution to Eq. \eqref{eq-stab-loss} stays close to the codespace at all times instead. However, as shown in Fig. \ref{fig-conv-loss}, despite the confinement of solutions to the codespace, the logical coherences inside this space itself decay.
\begin{figure}
    \includegraphics[width=\columnwidth]{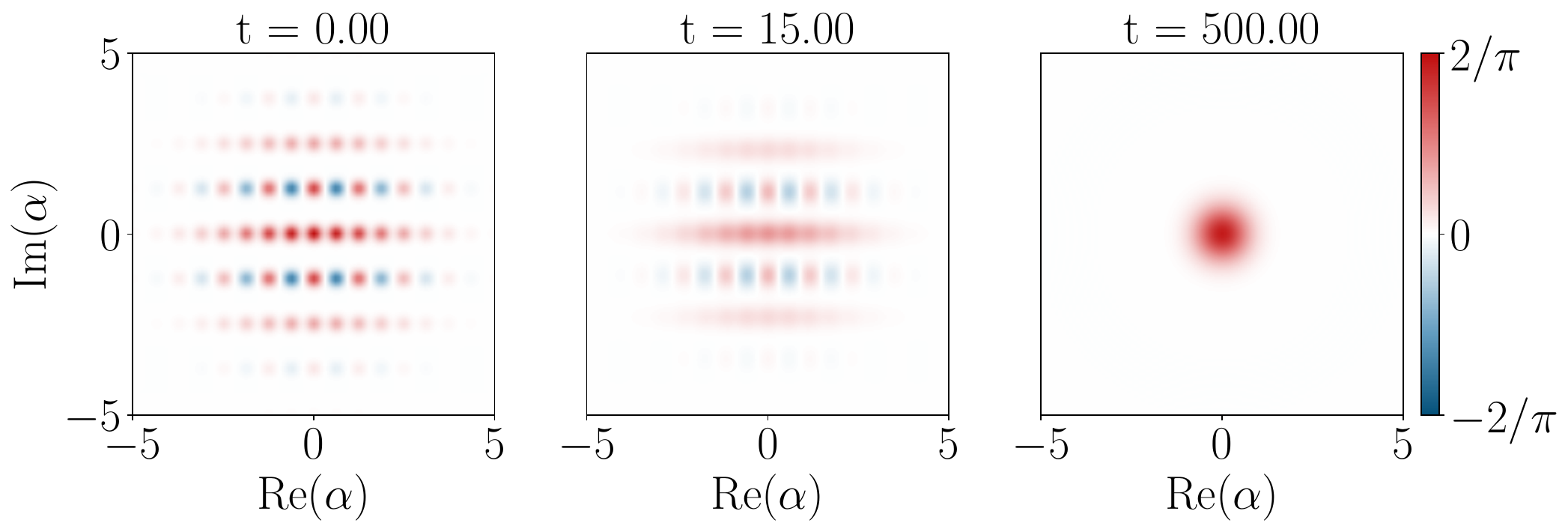}\\
    \includegraphics[width=\columnwidth]{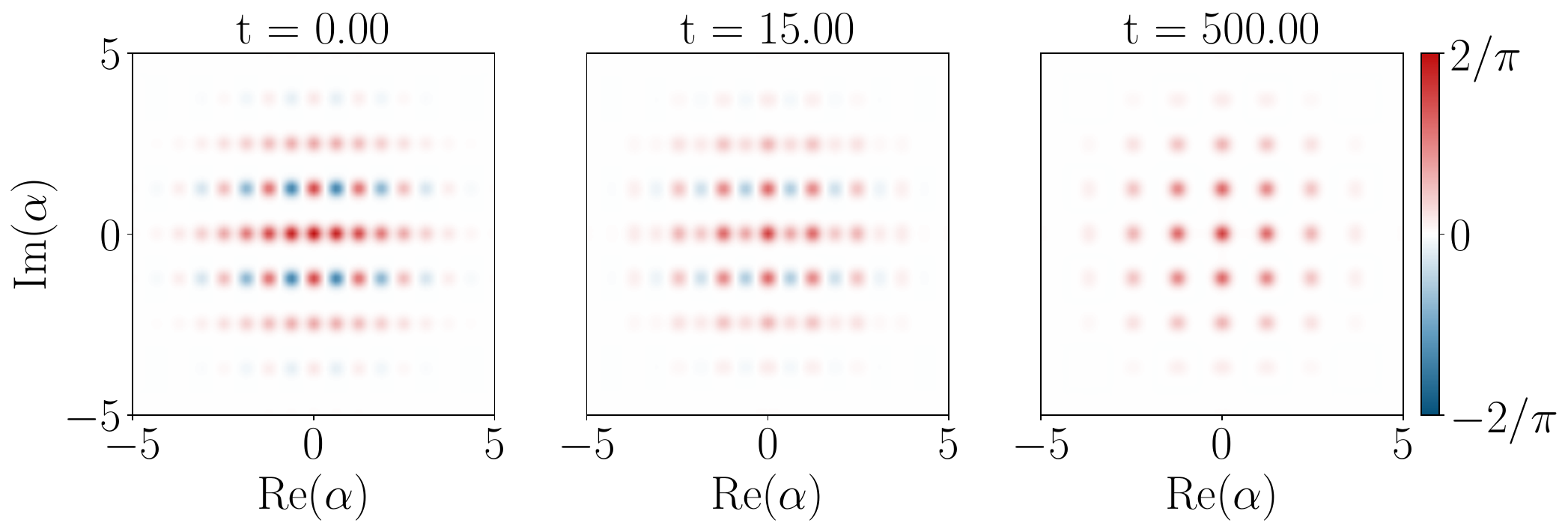}
    \caption{Evolution of an initial logical $|+X\rangle$ state in the presence of photon loss.
    {\bf Top.} Without stabilization, that is under $\frac{d}{dt}\orho = \kappa D[\oa](\orho)$, any state of a harmonic oscillator converges to vacuum (exponentially fast with a characteristic time $1/\kappa$).
    {\bf Bottom.} With the stabilization, that is under $\frac{d}{dt}\orho = D[\oM_1](\orho) + D[\oM_2](\orho) + \kappa D[\oa](\orho)$, the GKP space is protected but logical coherences inside that space are lost (decoherence to  the logical identity operator of the code, i.e.  the center of the logical Bloch sphere).
    Here, $\eta=\sqrt\pi$, $\epsilon=0.15$ and $\kappa=10^{-2}$.}
    \label{fig-conv-loss}
\end{figure}
\vspace{0.5em}

Following \cite{sellemStabilityDecoherence2023}, we define logical coordinate observables as
\begin{align}
    \oZ &= \sign(\cos(\etasq/2 \oq)) = \sign(\cos(\eta \oq)),\\
    \oX &= \sign(\cos(\etasq/2 \op)) = \sign(\cos(\eta \op)),\\
    \oY &= -i \oZ \oX.
\end{align}
\begin{figure}
    \centering
    \includegraphics[width=0.8\columnwidth]{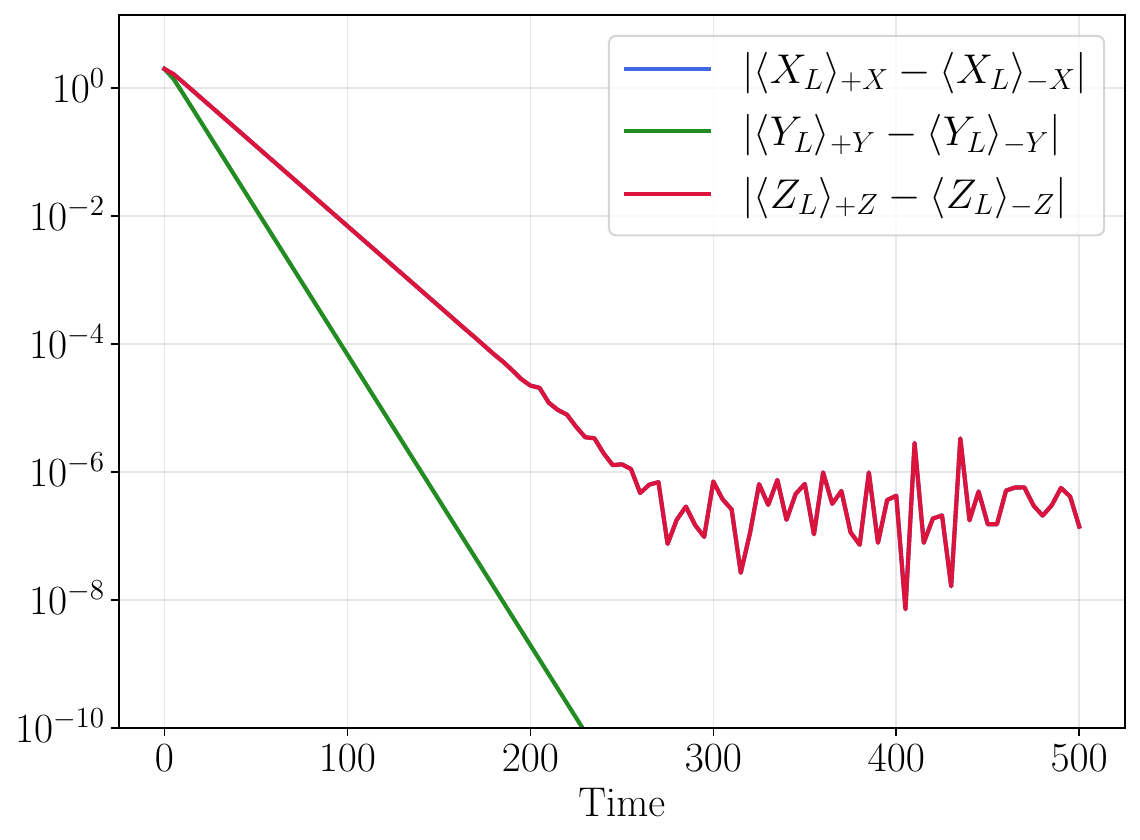}
    \caption{Decay of logical observables for $\epsilon=0.15$ and $\kappa=10^{-2}$. Following \cite{sellemDissipativeProtectionGKP2025}, we study the decay of the contrast between two opposite logical states along time (rather than values on a specific state) to easily get rid of final values on the steady-state of the dynamics. The different rate observed for $\oY$ compared to $\oX$ and $\oZ$ is a well-known feature of square GKP codes, due to the fact that $\oX$ and $\oZ$ approximate a displacement along orthogonal sides of a square of the grid while $\oY$  approximates a displacement along its diagonal hence longer.}
    \label{fig-decay-exp}
\end{figure}
\begin{figure}
        \centering
\includegraphics[width=\columnwidth]{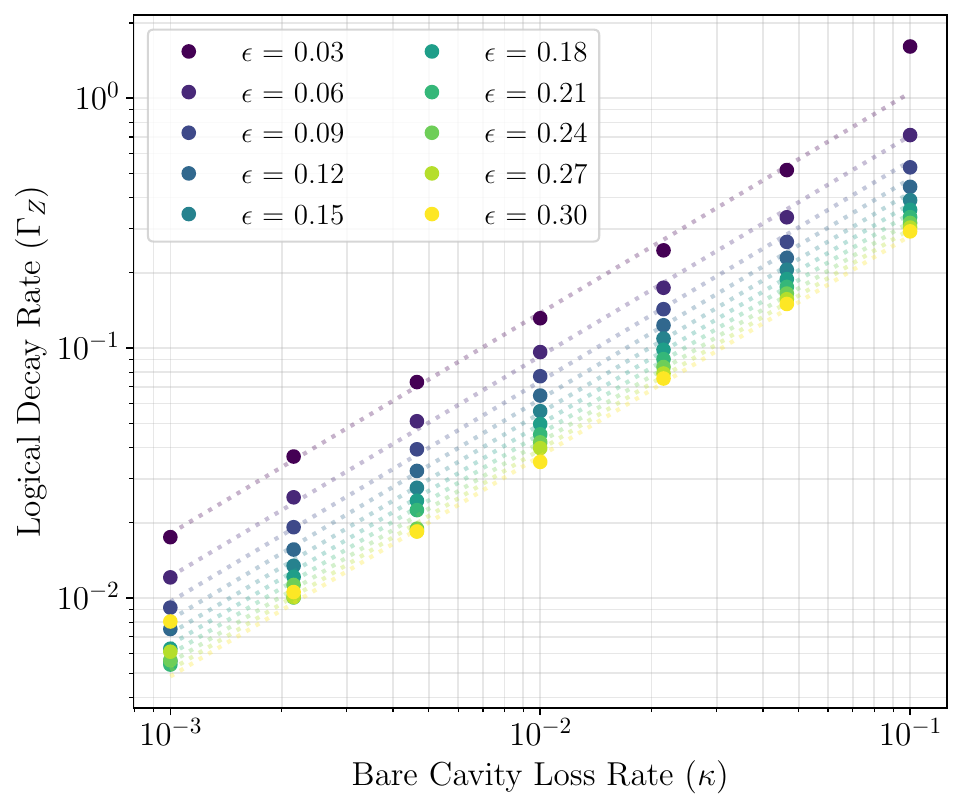}
\caption{Decay rate of the expectation value of $\oZ$ as a function of photon loss ($\kappa$) and regularizing parameter $\epsilon$. We attribute the dependence on $\epsilon$ to the fact that the mean photon number of finite-energy GKP states grows with $1/\epsilon$ \cite{sellemExponentialConvergenceDissipative2022a}. The dashed lines correspond to a numerical fit to a power law $\Gamma_Z = A\kappa^n/\epsilon^r$; we find $A\simeq 1$, $n\simeq 0.88$ and $r\simeq 0.57$. %
The range 
$\epsilon\in [0.03,0.3]$ spans roughly a mean photon number of $30$ down to $3$
 photons; current GKP experiments~\cite{campagne-ibarcqQuantumErrorCorrection2020, brockQuantumErrorCorrection2024} sit near the high-$\epsilon$ (low-photon) end of that range.}
\label{fig-scaling}
\end{figure}
The decay of the corresponding expectation values is well-captured by an exponential model, after a fast initial transient (see Fig. \ref{fig-decay-exp}). We can thus define logical decoherence rates $\Gamma_{X,Y,Z}$ associated to logical observables as the rate of this exponential decay. We plot in Fig. \ref{fig-scaling}
this decay rate as a function of both $\epsilon$ and $\kappa$.
It appears to be well captured by a power law of the form $\Gamma(\kappa,\epsilon) = A \kappa^n/ \epsilon^{r}$ with $A,n,r$ fitting parameters.
Crucially, this scaling is  qualitatively worse than that of the full four-dissipator dynamics in \cite{sellemDissipativeProtectionGKP2025}, where the logical error rate rather appeared well-captured by an exponential dependence of the form $\Gamma(\kappa,\epsilon) \propto \epsilon\,  e^{-1/\sigma(\kappa,\epsilon)}$ with $\sigma$ a linear expression of $\epsilon$ and $\kappa/\epsilon$.

\section{Stabilization of a GKP qunaught}
\label{sec-qunaught}
Apart from their application for quantum error correction, GKP grid states have also been considered as valuable resources in metrology.
They can allow to measure simultaneously the commuting modular observables derived from
$e^{2i\sqrt\pi \oq}$ and $e^{2i\sqrt\pi \op}$, or equivalently $\oq \textrm{ mod } \sqrt\pi$ and $\op \textrm{ mod } \sqrt \pi$,
circumventing the Heisenberg uncertainty principle attached to measurements of $\oq$ and $\op$.
These states can be generated through well-known universal control methods for harmonic oscillators, such as the ECD \cite{eickbuschFastUniversalControl2022} or SNAP methods \cite{heeresCavityStateManipulation2015}.

In contrast, the dynamics proposed in Eq. \eqref{eq-main-lindblad} (with the choice $\eta = \sqrt{\pi/2}$) allows for the stabilization
of such a grid-state in steady-state. As shown in Fig. \ref{fig-qunaught}, even in the presence of photon loss (additional dissipator $D[\oa]$ in the Lindblad equation), the fine periodic structure of the steady state that makes it relevant for metrology is preserved, while its contrast dies out.

\begin{figure}[htbp]
    \centering
    \includegraphics[width=0.5\textwidth]{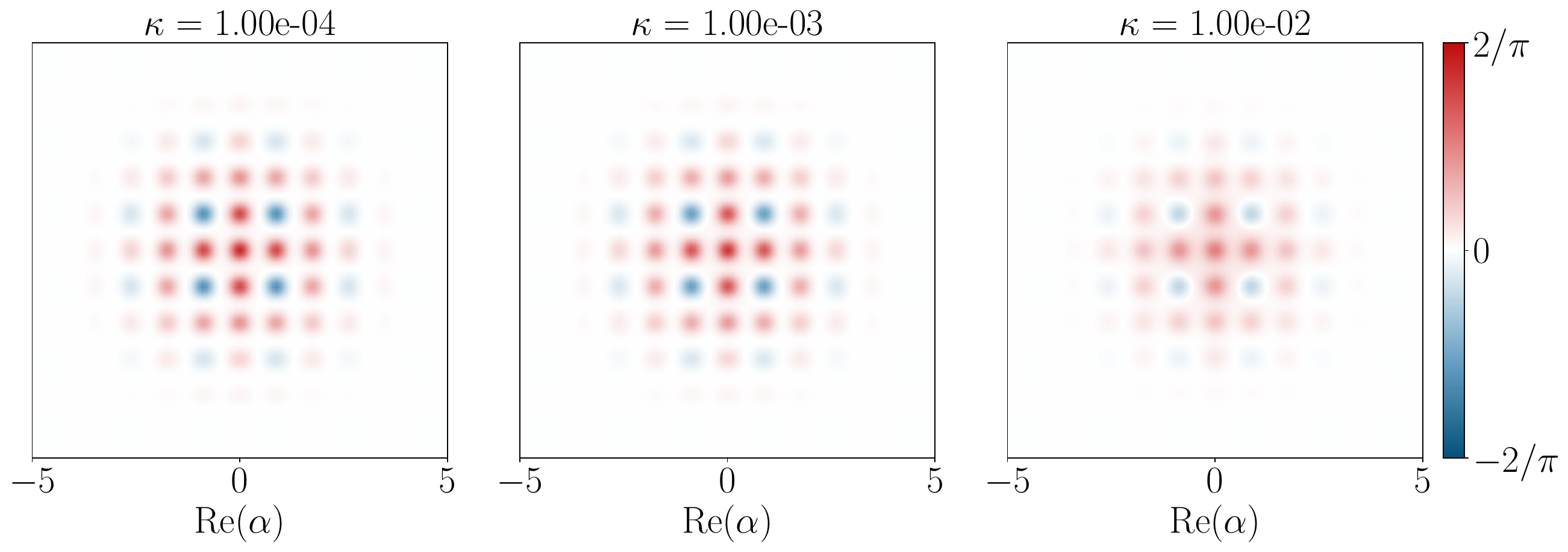}
    \caption{Steady-state of $\frac{d}{dt}\orho = D[\oM_1](\orho) + D[\oM_2](\orho) + \kappa D[\oa](\orho)$ as a function of the relative single-photon loss rate $\kappa$. Here $\eta = \sqrt{\pi/2}$ and $\epsilon=0.15$.}
    \label{fig-qunaught}
\end{figure}

\section{POSSIBLE IMPLEMENTATIONS}
\label{sec-implem}
In practice, the Lindblad equation \eqref{eq-main-lindblad}
models a quantum harmonic oscillator coupled to an unrealistic exotic bath. It can, however, be approximated through reservoir engineering methods as used in previous GKP proposals \cite{sellemDissipativeProtectionGKP2025} or for other bosonic codes \cite{lescanneExponentialSuppressionBitflips2020}.
In a nutshell, reservoir engineering methods rely on coupling a system of interest to a strongly damped auxiliary system; the complexity of bath engineering then transforms into the complexity of the coupling to engineer.

The complexity of these methods typically scales with the number of dissipation channels to engineer. For instance, each dissipation channel can be approximated by a dedicated ancilla system -- translating the number of dissipators to hardware complexity. Alternately, each necessary dissipator can be engineered stroboscopically, which corresponds to a Trotter decomposition of the evolution -- the achievable dissipation rate is then inversely proportional to the number of dissipators in this decomposition.
In this regard, stabilizing GKP states with half the number of dissipation channels can significantly alleviate implementation constraints.
\paragraph{circuitQED}
In the case of circuitQED, both the system of interest and the ancilla are harmonic oscillators, corresponding to electromagnetic modes typically at microwave frequencies. The implementation method proposed in \cite{sellemDissipativeProtectionGKP2025} for the stabilization of a GKP qubit with four dissipators (Eqs. \ref{eq-prx4}) could straightforwardly be adapted to our new dynamics with only two, since they directly correspond to their first two dissipators but with a different value of the parameter $\eta$. 
In addition, in their case $\eta$ is a function of the impedance of the mode of interest, which needed to be at twice the resistance quantum, a high value possibly  difficult to achieve in practice. Reducing $\eta$ alleviates this requirement.

\paragraph{Trapped ions}
Trapped ions also appear as a particularly appealing platform for GKP experiments, with several demonstrations of GKP stabilization reported in the literature. In this case, the harmonic oscillator is a motional degree of freedom of an ion (or possibly a collective mode of an ensemble of ions), and the auxiliary system is the spin degree of freedom of the same ion. 
The $\eta$ parameter in this paper depends on a quantity called the Lamb-Dicke parameter in that field \cite{winelandExperimentalIssuesCoherent1998}, which can depend both on the ion species and specificities of the trap.
A possibly interesting direction to avoid fine-tuning $\eta$ could be to resort to so-called Quantum Signal Processing (QSP) methods. Notably, in the context of trapped ions, the recent preprint \cite{mcgarryProgrammableQuantumSimulation2026} introduced a method to approximate nonlinear Hamiltonians based on their Fourier decomposition, a feature particularly appealing in our context where all relevant operators are perturbations of periodic operators.

\section{CONCLUDING REMARKS}
\label{sec-concl}
In conclusion, we have presented a simplified, two-dissipator Lindblad dynamics that approximately stabilizes finite-energy Gottesman-Kitaev-Preskill (GKP) states, and established its convergence properties. While numerical simulations show that the logical decoherence rates scale qualitatively worse than a previously proposed four-dissipator dynamics in presence of photon loss, it trades robustness for implementation simplicity. Both methods are also based on very similar approaches, so that one can expect technological developments used to realize the stabilization proposed here to also be directly useful for the more robust but complex four-dissipator approach. This two-dissipator approach thus appears as an easier near-term experimental objective to validate experimental developments needed for the full four-dissipator method.
We also showed that with a slight modification of parameters the same approach allows for the stabilization of GKP qunaught states for quantum metrology.

Note that we focused here on so-called square GKP states. We refer the reader to e.g. \cite{conrad2022gottesman,royerEncodingQubitsMultimode2022a} for a general theory of GKP qudits of arbitrary logical dimension on more general lattices. Similar stabilizing dynamics could be derived for these generalizations, in a similar fashion to this article where Lindblad operators are associated to logical stabilizers of the GKP code.

Finally, the analytical derivation of energy estimates in Sec. \ref{sec-stab} and decoherence rates in Sec. \ref{sec-log} relied on formal computations led as if the Hilbert space were finite-dimensional. We hope to leverage these formal results to develop a fully rigorous analysis of the system in
forthcoming publications.\\

Numerical simulations were run on a laptop GPU with double precision arithmetic (Nvidia Quadro RTX 3000 with 6GB of RAM), using the libraries jaxquantum \cite{jha2025jaxquantum} for the manipulation of GKP states and dynamiqs \cite{guilmin2025dynamiqs} for the numerical integration of Lindblad equations. The related source codes are available from  the corresponding author upon reasonable request.




\section*{APPENDIX}
\label{apdx}
\subsection{Derivation of Eq. \eqref{eq-bound-N}}
Here, we obtain a priori estimates by formal computations,
led as if the dimension of the underlying Hilbert space
were finite. We plan to exploit these estimates for a fully
rigorous mathematical analysis in future publications.

Formally, the evolution of the expectation value of an observable (that is a hermitian operator) $\oO$ on the solution $\orho_t$ to Eq. \eqref{eq-main-lindblad} can be obtained by duality through
\begin{align}
    \frac{d}{dt}\tr(\oO\orho_t) &= \tr(\mathcal L^*(\oO)\orho_t)\\
    &=\tr\left( D^*[\oM_1](\oO) \, \orho_t + D^*[\oM_2] (\oO)\,  \orho_t \right)
\end{align}
with $D^*[\oM](\oO) := \frac12\oM^\dag [\oO,\oM] + h.c.$
Using $\oN = \frac{\oq^2+\op^2-1}{2}$ allows deriving estimates on $\mathcal L^*(\oN)$. Below we'll repeatedly use the relation $[f(\oA),\oB] = f'(\oA) [\oA,\oB]$ valid for $f$ analytic and $\oA,\oB$ such that $[\oA, [\oA,\oB]] = 0$ \cite[Chapter 2]{cohen2019quantum}.  We'll also make use of the following operator inequalities, valid for any $\alpha,\beta>0$ and $r\in (0,1)$,
as well as their equivalent for $\op$:
\begin{align}
    -\oq \sin(2\eta \oq) &\leq |\oq|\\
    \cos(2\eta\oq) &\leq \oid\\
    -\alpha \oq^2 + \beta |\oq| &\leq -r\alpha\oq^2 + \frac{\beta^2}{4\alpha(1-r)}\oid. \label{eq-rtrick}
\end{align}
We find
\begin{align}
    [\oq^2, \oM_1] &= [\oq^2, \sin(\eta\oq) + i\epsilon\cos(\eta \oq) \op]\\
    &= i\epsilon\cos(\eta \oq) [\oq^2, \op] = -2\epsilon \oq \cos(\eta \oq)\\
    \oM_1^\dagger [\oq^2,\oM_1] &=
        -2\epsilon \oq\cos(\eta\oq)\sin(\eta\oq) +2 i\epsilon^2\op\oq\cos^2(\eta \oq)\\
        &= -\epsilon \oq\sin(2\eta\oq) + i\epsilon^2\op\oq(\cos(2\eta\oq) +\oid) 
\end{align}
hence
\begin{align}
    &\frac12 \oM_1^\dagger [\oq^2,\oM_1] + h.c.
     \notag\\
    &\quad = -\epsilon \oq\sin(2\eta\oq) + \frac {i\epsilon^2}2[\op,\oq(\cos(2\eta\oq)+\oid)]\\
    &\quad = -\epsilon \oq\sin(2\eta\oq) + \frac{\epsilon^2}2 (\cos(2\eta\oq)+\oid) - \eta \epsilon^2 \oq \sin(2\eta\oq)\\
    &\quad \leq \epsilon(1+\epsilon\eta) |\oq| + \epsilon^2\oid.
\end{align}
Similar but slightly tedious computations lead to
\begin{align}
    \frac12 \oM_1^\dag [\op^2, \oM_1] + h.c.
    \leq -2\epsilon\eta (1-\frac{\epsilon\eta}{2})\op^2 +3\epsilon\eta^3 \oid
\end{align}
so that all in all, for any $r\in (0,1)$:
\begin{align}
    D^*[\oM_1](\oN) &\leq
    -\epsilon\eta (1-\frac{\epsilon\eta}{2})\op^2 
    + \frac12\epsilon(1+\epsilon\eta)|\oq| \\
    &\qquad \qquad \qquad \qquad +\frac12(\epsilon^2 
    + 3\epsilon\eta^3)\oid \label{eq-step}\\
    &\leq -r \epsilon\eta (1-\frac{\epsilon\eta}{2})\op^2 + C(r,\epsilon,\eta)\oid
\end{align}
with $C$ depending on $r,\epsilon,\eta$ from combining the constant terms coming from using Eqs. \eqref{eq-rtrick} and \eqref{eq-step}.
Combining with the corresponding calculation for $\oM_2$, we finally obtain:
\begin{equation}
    \mathcal L^*(\oN) \leq -2 r \epsilon\eta (1-\frac{\epsilon\eta}{2})\oN + C'(r,\epsilon,\eta) \oid.
\end{equation}
Note that here $r$ is a free parameter in $(0,1)$ required in proof steps, but can be chosen arbitrarily close to $1$.

\subsection{Derivation of Eq. \eqref{eq:periodic_obs}}
We recall that $\oO_f = f(2\eta \oq)$ with $f$ a smooth $2\pi$-periodic function and $\eta=\sqrt{\pi}$. As $[\oM_2, \oO_f]=0$, 
$$\mathcal{L}^*(\oO_f)=\frac{1}{2} \oM_1^\dagger [\oO_f,\oM_1] + h.c.$$
The commutator $[\oO_f,\oM_1]$ reads
\begin{align}
[\oO_f,\oM_1] &= [f(2 \eta \oq),\sin(\eta \oq)] + i\epsilon [f(2\eta \oq),\cos(\eta \oq) \op] \nonumber\\
&=i\epsilon \cos(\eta \oq)  [f(2\eta \oq),\op] \nonumber\\
&=-\epsilon\cos(\eta \oq)\,2\eta f'(2\eta \oq).
\end{align}
where we used $\op=-i\partial_q$. Hence, it remains to compute
\begin{align*}
&\oM_1^\dagger[\oO_f,\oM_1] \\
&= (\sin(\eta \oq)-i\epsilon \op\cos(\eta \oq))(-\epsilon\cos(\eta \oq)\,\eta' f'(2\eta \oq)).
\end{align*}
Applying the product rule to the action of $\op$ on $\cos(\eta \oq)\,2\eta f'(2\eta q)$ produces two contributions: one from the derivative of the trigonometric factor and one from the derivative of $f'$, yielding a second-derivative term. Collecting terms gives
\begin{align*}
&\oM_1^\dagger[\oO_f,\oM_1]=
-\epsilon\eta\sin(2\eta  \oq) f'(2\eta \oq) \\
&+ 2\epsilon^2\eta\big(\cos^2(\eta \oq)\,2\eta f''(2\eta \oq)-2\eta\sin(\eta \oq)\cos(\eta \oq) f'(2\eta \oq)\big).\\
&= -(\epsilon \eta +2\epsilon^2 \eta^2) \sin(2 \eta \oq) f'(2\eta \oq) + 4 \epsilon^2 \eta^2 \cos( \eta \oq)^2 f''(2\eta \oq)
\end{align*}
This operator being already hermitian, we have shown Eq. \eqref{eq:periodic_obs}.

\section*{ACKNOWLEDGMENTS}
 We thank Philippe Campagne-Ibarcq and Baptiste Royer
 for useful discussions and comments.


\bibliographystyle{unsrt}
\bibliography{biblio}




\end{document}

%% file: symbolsmacros.tex
\newcommand{\orho}{{\boldsymbol \rho}}
\newcommand{\oA}{{\bf A}}
\newcommand{\oB}{{\bf B}}
\newcommand{\oE}{{\bf E}}

\newcommand{\oL}{{\bf L}}
\newcommand{\oN}{{\bf N}}
\newcommand{\oM}{{\bf M}}
\newcommand{\oO}{{\bf O}}
\newcommand{\oX}{{\bf X}}
\newcommand{\oY}{{\bf Y}}
\newcommand{\oZ}{{\bf Z}}

\newcommand{\oq}{{\bf q}}
\newcommand{\op}{{\bf p}}
\newcommand{\oa}{{\bf a}}
\newcommand{\oid}{{\bf Id}}
\newcommand{\etasq}{\eta_{\scriptscriptstyle \square}}
\DeclareMathOperator{\tr}{Tr}
\DeclareMathOperator{\sign}{Sign}

\newtheorem{lemma}{Lemma}
\newtheorem{theorem}{Theorem}
\newtheorem{corollary}{Corollary}

%% file: spectralgap.tex

For observables of the form $\oO_f = f(2\eta \oq)$ with $f$ a smooth $2\pi$-periodic function, one checks $[\oM_2, \oO_f]=0$ (using $2\eta^2=2\pi$), so that only $\oM_1$ contributes to the Heisenberg evolution. A direct computation, presented in the Appendix, yields
\begin{align}
\label{eq:periodic_obs}
    \mathcal{L}^*(\oO_f) &= -\big(\epsilon + 2\epsilon^2\eta\big) \eta\sin(2\eta \oq)\, f'(2\eta \oq) \nonumber\\
    &\quad + 4\epsilon^2 \eta^2 \cos^2(\eta \oq)\, f''(2\eta \oq),
\end{align}
where primes denote derivatives of $f$ with respect to its argument $\theta = 2\eta \oq$. This reduces the study to the one-dimensional operator
\begin{align}\label{eq:reduced-op}
    \mathcal{A} f(\theta) = \sin(\theta)\, f'(\theta) - \sigma\,(1+\cos\theta)\, f''(\theta),
\end{align}
up to a positive prefactor depending on $\epsilon$ and $\eta$.

\begin{lemma}[Symmetry]\label{lem:symmetry}
For $0 < \sigma < 2$, the operator $\mathcal{A}$ is symmetric and non-negative on $L^2(w)$ with weight $w(\theta) = (1+\cos\theta)^{1/\sigma - 1}$. More precisely, integration by parts gives
\begin{align}
    \langle f, \mathcal{A} g \rangle_w = \int_0^{2\pi} \sigma\,(1+\cos\theta)^{1/\sigma}\, f'(\theta)\, g'(\theta)\, \mathrm{d}\theta.
\end{align}
\end{lemma}
\begin{proof}
Write $w_2(\theta) = (1+\cos\theta)^{1/\sigma}$ and note that
\begin{align}
    w_2' = -\frac{\sin\theta}{\sigma}\,(1+\cos\theta)^{1/\sigma - 1} = -\frac{\sin\theta}{\sigma}\, w.
\end{align}
Integrating by parts,
\begin{align}
    \langle f, \mathcal{A} g \rangle_w &= \int_0^{2\pi} f\big[\sin\theta\, g' - \sigma\,(1+\cos\theta)\, g''\big] w\, \mathrm{d}\theta \nonumber\\
    &= -\sigma \int_0^{2\pi} f\, g''\, w_2\, \mathrm{d}\theta + \int_0^{2\pi} f\, g'\, \sin\theta\, w\, \mathrm{d}\theta \nonumber\\
    &= \sigma \int_0^{2\pi} f'\, g'\, w_2\, \mathrm{d}\theta + \sigma\int_0^{2\pi} f\, g'\, w_2'\, \mathrm{d}\theta \nonumber\\
    &\quad + \int_0^{2\pi} f\, g'\, \sin\theta\, w\, \mathrm{d}\theta.
\end{align}
The last two terms cancel since $\sigma\, w_2' + \sin\theta\, w = 0$.
\end{proof}

We now establish a weighted Poincar\'{e} inequality for $\mathcal{A}$. The difficulty lies in the degeneracy of $w_2$ at $\theta = \pi$, which we handle via localized Hardy-type estimates. We decompose the torus into $I_1^+ = [\pi, 3\pi/2]$, $I_1^- = [\pi/2, \pi]$, $I_1 = I_1^+ \cup I_1^-$ and $I_2 = [0,\pi/2] \cup [3\pi/2, 2\pi]$.

\begin{lemma}[Localized Hardy inequality]\label{lem:hardy}
For $0 < \sigma < 2$ and $g \in C^\infty([\pi, 3\pi/2])$ with $g(3\pi/2) = 0$,
\begin{align}
    \|g\|_{L^2(w,\, I_1^+)}^2 \leq \frac{16\sigma^2}{(2-\sigma)^2}\, \|g'\|_{L^2(w_2,\, I_1^+)}^2.
\end{align}
The same holds on $I_1^-$ for $g$ with $g(\pi/2) = 0$, by symmetry.
\end{lemma}
\begin{proof}
Set $\phi(\theta) = (1+\cos\theta)^{1/\sigma - 1/2}$ and differentiate:
\begin{align}
    \frac{d}{d\theta}\big(g^2 \phi\big) &= 2g\, g'\, \phi \nonumber\\
    &\quad - \Big(\frac{1}{\sigma} - \frac{1}{2}\Big)\sin\theta\, g^2\, (1+\cos\theta)^{1/\sigma - 3/2}.
\end{align}
Integrating over $[\pi, 3\pi/2]$, the boundary terms vanish ($\phi(\pi) = 0$ and $g(3\pi/2) = 0$), giving
\begin{align}
    &\Big(\frac{1}{\sigma} - \frac{1}{2}\Big) \int_\pi^{3\pi/2} g^2\, w\, \frac{(-\sin\theta)}{\sqrt{1+\cos\theta}}\, \mathrm{d}\theta \nonumber\\
    &\quad = -2\int_\pi^{3\pi/2} g\, g'\, (1+\cos\theta)^{1/\sigma - 1/2}\, \mathrm{d}\theta.
\end{align}
On $(\pi, 3\pi/2)$, $-\sin\theta / \sqrt{1+\cos\theta} \geq 1$, so the left-hand side is at least $(1/\sigma - 1/2)\|g\|_{L^2(w,I_1^+)}^2$. Applying Cauchy--Schwarz to the right-hand side,
\begin{align}
    \Big(\frac{1}{\sigma} - \frac{1}{2}\Big)\|g\|_{L^2(w,I_1^+)}^2 &\leq 2\,\|g\|_{L^2(w,I_1^+)}\,\|g'\|_{L^2(w_2,I_1^+)},
\end{align}
whence
\begin{align}
    \|g\|_{L^2(w,I_1^+)} \leq \frac{4\sigma}{2-\sigma}\,\|g'\|_{L^2(w_2,I_1^+)}.
\end{align}
\end{proof}

\begin{lemma}[Poincar\'{e} estimate on $I_2$]\label{lem:poincare}
On $I_2$ the weights $w$ and $w_2$ are bounded above and below by positive constants. Hence, for $g$ with $g(3\pi/2) = 0$ (or $g(\pi/2)=0$),
\begin{align}
    \|g\|_{L^2(w,\,I_2)}^2 &\leq C_\sigma\, \|g'\|_{L^2(w_2,\,I_2)}^2, \\
    |g(\pi/2) - g(3\pi/2)|^2 &\leq C_\sigma\, \|g'\|_{L^2(w_2,\,I_2)}^2.
\end{align}
\end{lemma}
~\\

\begin{theorem}[Spectral gap]\label{thm:gap}
For $0 < \sigma < 1/2$, there exists $C_\sigma > 0$ such that for any smooth $2\pi$-periodic $f$ with $\int_0^{2\pi} f(\theta)\, w(\theta)\, \mathrm{d}\theta = 0$,
\begin{align}
    \|f\|_{L^2(w)}^2 \leq C_\sigma \int_0^{2\pi} (1+\cos\theta)^{1/\sigma}\, |f'(\theta)|^2\, \mathrm{d}\theta.
\end{align}
\end{theorem}

\begin{proof}
We decompose the norm and subtract boundary values to create functions vanishing at the interface points:
\begin{align}
    \|f\|_{L^2(w)} &\leq \|f - f(3\pi/2)\|_{L^2(w,\,I_1^+)} \nonumber\\
    &\quad + \|f - f(\pi/2)\|_{L^2(w,\,I_1^-)} \nonumber\\
    &\quad + \|f\|_{L^2(w,\,I_2)} \nonumber\\
    &\quad + \sqrt{\|w\|_{L^1(I_1)}}\big(|f(3\pi/2)| + |f(\pi/2)|\big).
\end{align}
The first two terms are controlled by the Hardy inequality (Lemma~\ref{lem:hardy}). For the $I_2$ term, we write
\begin{align}
    \|f\|_{L^2(w,\,I_2)} &\leq \|f - f(3\pi/2)\|_{L^2(w,\,I_2)} \nonumber\\
    &\quad + \sqrt{\|w\|_{L^1(I_2)}}\,|f(3\pi/2)|,
\end{align}
and the Poincar\'{e} bound (Lemma~\ref{lem:poincare}) controls the first term. The same lemma bounds $|f(3\pi/2) - f(\pi/2)|$, so it remains to control a single point value, say $|f(\pi/2)|$.

The zero-mean condition gives
\begin{align}
    \|w\|_{L^1}\,|f(\pi/2)| &= \Big|\int_0^{2\pi}\big(f(\pi/2) - f(\theta)\big)\, w(\theta)\, \mathrm{d}\theta\Big| \nonumber\\
    &\leq \sqrt{\|w\|_{L^1}}\,\|f - f(\pi/2)\|_{L^2(w)}.
\end{align}
The right-hand side is split over $I_1^-$ (Hardy), $I_2$ (Poincar\'{e}), and $I_1^+$ where we use
\begin{align}
    &\|f - f(\pi/2)\|_{L^2(w,\,I_1^+)} \nonumber\\
    &\quad \leq \|f - f(3\pi/2)\|_{L^2(w,\,I_1^+)} \nonumber\\
    &\quad\quad + |f(3\pi/2) - f(\pi/2)|\,\sqrt{\|w\|_{L^1(I_1^+)}},
\end{align}
with both terms already estimated. Collecting all bounds yields the result.
\end{proof}

\begin{corollary}[Convergence of periodic observables]\label{cor:exp-conv}
    Recall that $\sigma = \frac{2\epsilon\eta}{1+2\epsilon\eta}$, so $0 < \sigma < 1/2$ is equivalent to $\epsilon\eta < 1/2$. Let $C_\sigma$ be the constant from Theorem~\ref{thm:gap}. For any smooth $2\pi$-periodic function $f$ with $\int_0^{2\pi} f(\theta)\, w(\theta)\, \mathrm{d}\theta = 0$, the Heisenberg evolution of $\oO_f = f(2\eta \oq)$ satisfies
    \begin{align}
        \|f_t\|_{L^2(w)} \leq e^{-\gamma_\sigma\, t}\, \|f\|_{L^2(w)},\, \gamma_\sigma = \frac{(\epsilon + 2\epsilon^2 \eta)\eta\,\sigma}{C_\sigma} = \frac{2\epsilon^2 \eta^2}{C_\sigma},
    \end{align}
    where $f_t$ denotes the function such that $e^{t\mathcal{L}^*}\oO_f = \oO_{f_t}$.

    More generally, for any smooth $2\pi$-periodic $f$, the observable $e^{t\mathcal{L}^*}\oO_f$ converges exponentially to $\bar f_w \cdot \oid$ at rate $\gamma_\sigma$, where $\bar f_w = \int_0^{2\pi} f\, w\, \mathrm{d}\theta \,/\, \int_0^{2\pi} w\, \mathrm{d}\theta$ is the weighted mean of $f$.
\end{corollary}